\documentclass[onecolumn]{IEEEtran}
\usepackage[margin=1.5in]{geometry}
\usepackage{amsmath,amsthm,thmtools,mathtools,amsfonts}
\usepackage{tikz}
\usepackage{qtree}
\usepackage{url,hyperref}
\usepackage{comment}
\newcommand{\seq}[2][]{\boldsymbol{#2}_{#1}}
\newcommand{\reducesto}{\,{\xrightarrow{dd}}\allowbreak\,}

\usepackage{mleftright}\mleftright
\theoremstyle{plain}
\newtheorem{thm}{\protect\theoremname}
\theoremstyle{plain}
\newtheorem{lem}[thm]{\protect\lemmaname}
\theoremstyle{plain}
\newtheorem{cor}[thm]{\protect\corollaryname}
\theoremstyle{remark}

\theoremstyle{definition}
\newtheorem{exm}[thm]{\protect\examplename}
\theoremstyle{definition}
\newtheorem{defn}[thm]{\protect\definitionname}
\theoremstyle{plain}
\newtheorem{claim}[thm]{\protect\claimname}
\theoremstyle{plain}
\newtheorem*{claim*}{\protect\claimname}

\providecommand{\claimname}{Claim}
\providecommand{\corollaryname}{Corollary}
\providecommand{\definitionname}{Definition}
\providecommand{\lemmaname}{Lemma}
\providecommand{\remarkname}{Remark}
\providecommand{\theoremname}{Theorem}
\providecommand{\examplename}{Example}

\begin{document}

\newcommand{\dupsto}{\,{\xrightarrow{d}}\allowbreak\,}
\renewcommand{\root}{\operatorname{root}}

\title{Duplication Distance to the Root\\ for Binary Sequences}
\date{}
\author{Noga~Alon, Jehoshua Bruck,~\IEEEmembership{Fellow,~IEEE, }Farzad~Farnoud,~\IEEEmembership{Member,~IEEE,} and~Siddharth~Jain,~\IEEEmembership{Student Member,~IEEE}%
\thanks{Noga Alon is with the Schools of Mathematics and Computer Science, Tel Aviv University, Tel Aviv 6997801, Israel, Email: \href{mailto:nogaa@post.tau.ac.il}{nogaa@post.tau.ac.il}.}%
\thanks{Jehoshua Bruck is with the Electrical Engineering Department, California Institute of Technology, Pasadena, CA, 91125, Email: \href{mailto:bruck@caltech.edu}{bruck@caltech.edu}.}%
\thanks{Farzad Farnoud is with the Department of Electrical and Computer Engineering, University of Virginia, Charlottesville, VA, 22903, Email: \href{mailto:farzad@virginia.edu}{farzad@virginia.edu}. He was with the Electrical Engineering Department, California Institute of Technology.}%
\thanks{Siddharth Jain is with the Electrical Engineering Department, California Institute of Technology, Pasadena, CA, 91125, Email: \href{mailto:sidjain@caltech.edu}{sidjain@caltech.edu}.}
\thanks{This paper was presented in part at 2016 IEEE International Symposium on Information Theory in Barcelona, Spain.}}

\maketitle
\begin{abstract}
We study the tandem duplication distance between binary sequences and their roots. In other words, the quantity of interest is the number of tandem duplication operations of the form $\seq x = \seq a \seq b \seq c \to \seq y = \seq a \seq b \seq b \seq c$, where $\seq x$ and $\seq y$ are sequences and $\seq a$, $\seq b$, and $\seq c$ are their substrings,  needed to generate a binary sequence of length $n$ starting from a square-free sequence from the set $\{0,1,01,10,010,101\}$. This problem is a restricted case of finding the duplication/deduplication distance between two sequences, defined as the minimum number of duplication and deduplication operations required to transform one sequence to the other. We consider both exact and approximate tandem duplications. For exact duplication, denoting the maximum distance to the root of a sequence of length $n$ by $f(n)$, we prove that $f(n)=\Theta(n)$. For the case of approximate duplication, where a $\beta$-fraction of symbols may be duplicated incorrectly, we show that the maximum distance has a sharp transition from linear in $n$ to logarithmic at $\beta=1/2$. We also study the duplication distance to the root for sequences with a given root and for special classes of sequences, namely, the de Bruijn sequences, the Thue-Morse sequence, and the Fibbonaci words. The problem is motivated by genomic tandem duplication mutations and the smallest number of tandem duplication events required to generate a given biological sequence.
\end{abstract}

\section{Introduction}

The genome of every organism is subject to mutations resulting from imperfect genome replication  as well as environmental factors. These mutations include \emph{tandem duplications}, which  create \emph{tandem repeats} by duplicating a  substring and inserting it adjacent to the original (e.g., $A\underline{CG}T\to A\underline{CGCG}T$);  and \emph{point mutations} or \emph{substitutions}, which substitute one base in the sequence by another (e.g., $A\underline{C}GT\to A\underline{T}GT$). Gaining a better understanding of mutations that modify genomes --thereby creating the variety needed for natural selection-- is helpful in many fields including phylogenomics, systems biology, medicine, and bioinformatics. 

One aspect of this task is the study of how genomic sequences are generated through mutations. In this paper, we focus on tandem duplication mutations and tandem repeats, which form about $3\%$ of the human genome~\cite{lander2001initial}, and study the minimum number of mutation events that can create a given sequence. More specifically, we define distance measures between pairs of sequences based on the number of exact or approximate tandem duplications that are needed to transform one sequence to the other. We then study the distances between sequences of length $n\in\mathbb N$ and their roots, i.e., the shortest sequences from which they can be obtained via these operations.

Formally, a (\emph{tandem}) \emph{repeat of length $h$} in a sequence is two identical adjacent blocks, each consisting of $h$ consecutive elements. For example, the sequence $12\underline{134134}51$ contains the repeat $134134$ of length $3$. A repeat of length $h$ may be created through a (\emph{tandem}) \emph{duplication of length $h$}, e.g., $1213451\dupsto1213413451$, where $\dupsto$ denotes a duplication operation. On the other hand, a repeat may be removed through a (\emph{tandem}) \emph{deduplication of length $h$}, i.e., by removing one of the two adjacent identical blocks, e.g., $1213413451\reducesto1213451$. 

The \emph{duplication/deduplication distance} between two sequences $\seq x$ and $\seq y$ is the smallest number of duplications and deduplications that can turn $\seq x$ into $\seq y$ (to denote sequences we use bold symbols). We set the distance to infinity if the task is not possible, for example, if $\seq x=1$ and $\seq y = 0$.

For two sequences $\seq x$ and $\seq y$, if $\seq y$ can be obtained from $\seq x$ through duplications, we say that $\seq x$ is an \emph{ancestor} of $\seq y$ and that $\seq y$ is a \emph{descendant} of $\seq x$.  An ancestor $\seq x$ of $\seq y$ is a \emph{root} of $\seq y$, denoted $\seq x = \root(\seq y)$, if it is \emph{square-free}, i.e., it does not contain a repeat. We define the  \emph{duplication distance} between two sequences as the minimum number of duplications required to convert the shorter sequence to the longer one.\footnote{Note that using the term distance here is a slight abuse of notation as the duplication distance does not satisfy the triangle inequality.} This distance is finite if and only if one sequence is the ancestor of the other. This paper is focused on finding bounds on the duplication distance of sequences to their roots. From an evolutionary point of view, the duplication distance between a sequence and its root is of interest since it corresponds to a likely path through which a root may have evolved into a sequence present in the genome of an organism.

 Our attention here is limited to binary sequences for the sake of simplicity, since for the binary alphabet, the root of every sequence is unique and belongs to the set $\{0,1,01,10,010,101\}$. Specifically, the roots of $0^n$ and $1^n$, $n\in \mathbb N$, are $0$ and $1$, respectively. For every other binary sequence $\seq s$ of length $n$, the root of $\seq s$ is the sequence in the set $\{01,10,010,101\}$ that starts and ends with the same symbols as $\seq s$. For example, the root of $\seq s=1001011$ is $101$ since 
\[101\dupsto\underline{1010}1\dupsto1010\underline{11}\dupsto1\underline{00}1011=\seq s.\]  
A \emph{run} in a sequence is a maximal substring consisting of one or more copies of a single symbol. Through duplication, we can generate every binary sequence from its root by first creating the correct number of runs of appropriate symbols. For example, since $\seq s=1001011$ has $5$ runs, the first being a run of the symbol $1$, we first generate $10101$ through duplication. It is not difficult to see that this is always possible. The next step is then to extend each run so that it has the appropriate length. 

In the proofs in the paper, it is often helpful to think of the distance to the root in terms of converting a sequence to its root via a sequence of  deduplications, e.g. the sequence $\seq s$ above can be \emph{deduplicated to} its root as 
 $$\seq s =1\underline{00}1011\reducesto1010\underline{11}\reducesto\underline{1010}1\reducesto101=\root(\seq s).$$
 
 We note that a celebrated result by Thue from 1906~\cite{thue1906} states that for alphabets of size $\ge3$, there is an infinite square-free sequence. Thus, in contrast to the binary alphabet, the set of roots for such alphabets is infinite since each substring of Thue's sequence is square-free.
 
For a binary sequence $\seq s$, let $f(\seq s)$ denote the duplication distance between $\seq s$ and its root and let $f(n)$ be the maximum of $f(\seq s)$ for all sequences $\seq s$ of length $n$. Table~\ref{tab:f(n)}, which was obtained through computer search, shows the values of $f(n)$ for $1\le n\le32$.

In this paper, we provide bounds on $f(\seq s)$ and on $f(n)$. We also consider a variation of the duplication distance, referred to as the \emph{approximate-duplication distance}, where the duplication process is imprecise and so the inserted block is not necessarily an exact copy. Specifically, the \emph{$\beta$-approximate-duplication distance} between two sequences $\seq x$ and $\seq y$ is the smallest number of duplications that can turn the shorter sequence into the longer one, where each duplication may produce a block that differs from the original in at most a $\beta$-fraction of positions. This distance between $\seq s$ and any of its roots is denoted by $f_\beta(\seq s)$ and the maximum of $f_\beta(\seq s)$ over all sequences $\seq s$ of length $n$ is denoted by $f_\beta(n)$. We provide bounds on $f_\beta(n)$ and in particular show that there is a sharp transition in the behavior of $f_\beta$ at $\beta=1/2$.

\begin{table}
\centering{}%
\begin{tabular}{|c||c|c|c|c|c|c|c|c|c|c|c|c|c|c|c|c|}
\hline 
$n$ & 1 & 2 & 3 & 4 & 5 & 6 & 7 & 8 & 9 & 10 & 11 & 12 & 13 & 14 & 15 & 16\tabularnewline
\hline 
$f(n)$ & 0 & 1 & 2 & 2 & 3 & 4 & 4 & 5 & 6 & 6 & 7 & 7 & 8 & 8 & 9 & 9\tabularnewline
\hline 
\hline 
$n$ & 17 & 18 & 19 & 20 & 21 & 22 & 23 & 24 & 25 & 26 & 27 & 28 & 29 & 30 & 31 & 32\tabularnewline
\hline 
$f(n)$ & 10 & 10 & 11 & 11 & 11 & 12 & 12 & 12 & 13 & 13 & 13 & 14 & 14 & 14 & 15 & 15\tabularnewline
\hline 
\end{tabular}
\caption{$f(n)$ for $1\le n\le32$\label{tab:f(n)}.}
\end{table}

Since each binary sequence has a unique root in the set $\left\{ 0,1,01,10,010,101\right\}$, the set of sequences can be partitioned based on their roots. In the paper, we also study the duplication distance to the root for sequences based on the part they belong to, that is, we consider $f_{\seq{\sigma}}(n)$ for $\seq{\sigma}\in\left\{ 0,1,01,10,010,101\right\}$, where $f_{\seq{\sigma}}(n)=\max\left\{ f(\seq{s}):\root(\seq{s})=\sigma,|\seq s|=n\right\}$.

\vspace{.3cm}
The rest of the paper is structured as follows. In the next two subsections, we summarize the results of the paper and describe the related work. Then, in Section \ref{sec:genBounds}, we prove the bounds on $f(n)$ and discuss some variants, as well as special classes of sequences. In Section~\ref{sec:Lsystems}, 
 we discuss duplication distance for special class of sequence generating systems called
Lindenmayer Systems. In Sections~\ref{sec:mismatch} and~\ref{sec:stringSystems}, we study the approximate-duplication distance to the root and the duplication distance for different roots, respectively. Finally, several open problems and possible future directions are presented in Section~\ref{sec:conc}.

\subsection{Results}

In this subsection, we present the main results of the paper. The proofs, unless they are very short, are postponed to later sections.

Suppose the root of $\seq s$ is $\seq{\sigma}\in\left\{ 0,1,01,10,010,101\right\} $. It is easy to see that 
\[
\log\frac{|\seq{s}|}{|\seq{\sigma}|}\le f(\seq{s})\le|\seq{s}|.
\]
While the above lower bound is tight in the sense that there exist $\seq{\sigma}$ and $\seq{s}$ that satisfy it with equality, e.g., $\seq s=0^{2^k}$ and $\seq\sigma=0$, we show there is a positive constant $c$ such that for most sequences of length $n$, the duplication distance to the root is bounded below by $cn$. We also improve the upper bound.
\begin{restatable}{thm}{thmbounds}
	\label{thm:bounds}The limit $\lim_{n\to\infty}f(n)/n$
	exists and
	\[
	0.045\le \lim_{n\to\infty}\frac{f(n)}n\le\frac{2}{5}\ \cdot
	\]
	Furthermore, for sufficiently large n, $f(\seq{s})\ge 0.045 n$
	for all but an exponentially small fraction of sequences $\seq{s}$ of length~$n$; and $f(n)\le2n/5+15$. 
\end{restatable}
Although the linear lower bound on the duplication distance to the root holds for almost all sequences, finding a specific family of sequences that satisfy it appears to be difficult. The next lemma and its corollary give the best known construction for a family with large distance to the root, namely, this family achieves distance $\Omega(n/\log n)$.

\begin{lem}
\label{lem:SeqDepBound}Consider a sequence $\seq{s}$ and a positive integer $k\ge4$, and let $K(\seq{s})$ denote the number of distinct $k$-mers (sequences of length $k$) occurring in $\seq s$. We have
\[
f(\seq{s})\ge\frac{K(\seq{s})}{k-1}\ \cdot
\]
\end{lem}
\begin{proof} For two sequences $\seq x=\seq t\seq u\seq u\seq v$ and $\seq y=\seq t\seq u\seq v$, we have $K(\seq y)\ge K(\seq x)-(k-1)$, since the only case in which a $k$-mer occurs in $\seq x$ but not in $\seq y$ is when the only instance of that $k$-mer intersects both copies of $\seq u$ in $\seq x$. There are at most $k-1$ $k$-substrings of $\seq x$ that intersect both copies of $\seq u$. Finally, no root contains a $k$-mer for $k\ge4$.
\end{proof}

A \emph{binary De Bruijn sequence}~\cite{debruijn1946combinatorial} of order $k$ is a binary sequence of length $n=2^k$ that when viewed cyclically contains every possible binary sequence of length $k$ as a substring exactly once. For example, $0011$ and $00010111$ are De Bruijn sequences of order 2 and order 3, respectively. A binary De Bruijn sequence of order $k$ and length $n=2^k$ has precisely $n-k+1$ distinct $k$-mers. Hence, we have the following corollary.
\begin{cor}
	For any binary De Bruijn sequence $\seq{s}$ of order $k$ (which has length
	$n=2^{k}$), we have 
	\[
	f(\seq{s})\ge\frac{n-\log_2 n}{\log_2 n}\ \cdot
	\]
\end{cor}
It is worth noting that using the same technique as the proof of $f(n)=\Omega(n)$ in Theorem~\ref{thm:bounds}, and the fact that there are at least $\frac{2^{n/2}}{n}$ De Bruijn sequences of length $n$ when $n$ is a power of two,\footnote{If De Bruijn seqences are defined cyclically as opposed to linearly, there are exactly $\frac{2^{n/2}}{n}$ De Bruijn sequences of length $n$} one can show that the largest duplication distance for De Bruijn sequences grows linearly in their length.

A question arising from observing that $f(n)=\Theta(n)$ is that how does allowing mismatches in the duplication process affect the  distance to the root. In particular, for what values of $\beta$, is $f_\beta(n)$ linear in $n$ and for what values is it logarithmic? The next theorem establishes that there is a sharp transition at $\beta=1/2$.
\begin{restatable}{thm}{thmbetalessthanhalf}\label{thm:betalessthanhalf}
If $\beta<1/2$, then there exists a constant $c=c(\beta)>0$
such that 
\[
f_{\beta}(n)\ge cn.
\]
Furthermore, if $\beta>1/2$, for any constant $C>\left\lceil \frac{2\beta+1}{2\beta-1}\right\rceil ^{2}$ and sufficiently large $n$,
\[
f_{\beta}(n)\le C\ln n.
\]
\end{restatable}

Finally, we establish that the limit of $\frac{f(n)}{n}$ is the same if we consider only sequences with root $10$ or only sequences with root $101$.
\begin{restatable}{thm}{thmdiffrentroots}
The limits $\lim_{n}\frac{f_{10}(n)}{n}$ and $\lim_{n}\frac{f_{101}(n)}{n}$ exist and are equal to $\lim_n\frac{f(n)}{n}$.\end{restatable}

\subsection{Related Work}\label{sub:relwork}
Tandem duplications and repeats in sequences have been studied from a variety of points of view. One of the most relevant to this work is the study of estimating the tandem duplication history of a given sequence, i.e., a sequence of duplication events that may have generated the sequence, see e.g., \cite{benson1999reconstructing,tang2002zinc,gascuel2005reconstructing}. While the afforementioned works study the problem from an algorithmic point of view, in this paper, we are focused on extremal distance values for binary sequences. Furthermore, \cite{tang2002zinc,gascuel2005reconstructing} have a more restrictive duplication model than that of the present paper.

Another aspect, the study of the ability of duplication mutations to generate diversity, has been recently investigated from an information-theoretic point of view~\cite{farnoud2016capacity,jain2015capacity}
. In particular, \cite{farnoud2016capacity} models sequences generated from a starting ``seed'' through different types of duplications as sequence systems and studies their \emph{capacity} and \emph{expressiveness}. The notion of capacity quantifies the ability of the systems to generate diverse families of sequences, and expressiveness is concerned with determining whether every sequence can be generated as a substring of another sequence, if not independently. The results in~\cite{farnoud2016capacity,jain2015capacity} include lower bounds on the capacity of tandem duplications and establishing that certain systems have nonzero capacity.  The aforementioned works focus on the possibility of generating sequences and do not consider the number of duplication steps it takes to do so for any given sequence, which is the subject of the current paper. 

Finally, we mention that the stochastic behavior of certain duplication systems has been studied in~\cite{elishco2016capacity,farnoud2015stochastic}, and error-correcting codes for combating duplication errors have been introduced in~\cite{jain2016duplication}.

\section{Bounds on $f(n)$\label{sec:genBounds}}
\thmbounds*
The lower bound of Theorem~\ref{thm:bounds} is proved with the help of Theorem~\ref{t21}, and its upper bound uses Theorem~\ref{thm:ub}. These theorems are stated next.
\begin{thm}
\label{t21}For $0<\alpha<1$, consider the set of the $\left\lfloor 2^{n\alpha}\right\rfloor $ sequences of length $n$ with the smallest duplication distance to the root and let $F_{\alpha}$ be the maximum duplication distance to the root for a sequence in this set. Then
\begin{equation}
6n\sum_{f=1}^{F_{\alpha}}\binom{n+f}{f}\binom{2n+f}{f}\binom{2n+f+2}{f}2^{f}\ge2^{n\alpha}-1.\label{eq:lb2}
\end{equation}
\end{thm}

Before stating the proof, we present some background, definitions, and a useful claim, as well as a simpler but weaker result.

Recall that if the sequence $\seq{s}=s_{1}s_{2}\dotsm s_{m}$ contains a repeat, then omitting one of the two blocks of this repeat to obtain a new sequence is called a deduplication. We also refer to the resulting sequence $\seq{s}'$ as a deduplication of $\seq{s}$, and write $\seq s\reducesto\seq s'$. A \emph{deduplication process} for a binary sequence $\seq{s}$ is a sequence of sequences $\seq{s}=\seq[0]{s}\reducesto\seq[1]{s}\reducesto\seq[2]{s}\reducesto\dotsm\reducesto\seq[f]{s}=\root(\seq s)$, where each $\seq[i+1]{s}$ is a deduplication of $\seq[i]{s}$ and the final sequence $\seq[f]{s}$ is the (square-free) root of $\seq s$. The \emph{length} of the deduplication process above is $f$, that is, the number of deduplications in it. A deduplication of $\seq{s}$ is an $(i,h)$-\emph{step} if $i$ is the starting position of (the first block) of a repeat of length $h$ and one of the blocks of this repeat is omitted. For example, if $\seq{s}=123\underline{134}13451$, a $(4,3)$-step produces $\seq s'=12313451$. A deduplication process of length $f$ of a sequence $\seq s$  can be described by a sequence of pairs $(i_{t},h_{t})_{t=1}^{f}$, where step number $t$ is an $(i_{t},h_{t})$-step. It is not difficult to check that knowing the final sequence in the process, and knowing all the pairs $(i_{t},h_{t})$ of deduplications in the process, in order, we can reconstruct the original sequence $\seq s$. 

From the preceding discussion, each binary sequence $\seq s$ can be encoded as the pair $\left(\seq \sigma,(i_{t},h_{t})_{t=1}^{f(\seq s)}\right),$ where $\seq \sigma$ is the root of $\seq s$ and $(i_{t},h_{t})_{t=1}^{f(\seq s)}$ a deduplication process of $\seq s$. Since there are only $6$ possibilities for $\seq \sigma$, and less than $n^{2}$ possibilities for each pair $(i_{t},h_{t})$, if $F=f(n)$, then
\begin{equation}
6\sum_{f=1}^{F}\left(n^{2}\right)^{f}\geq2^{n},\label{eq:n/logn}
\end{equation}
which implies that $F=f(n)=\Omega(n/\log n)$. 

In the aforementioned encoding, several deduplication processes may map to the same sequence. We improve upon~(\ref{eq:n/logn}) by defining deduplication processes of a special form that remove some of the redundancy, and by doing so, we obtain~(\ref{eq:lb2}), which will lead to the linear lower bound of Theorem~\ref{thm:bounds}.
\begin{defn}
A deduplication process $\seq s=\seq[0]s\reducesto \seq[1]s\reducesto \seq[2]s\reducesto\dotsm\reducesto \seq[f]s=\root(\seq s)$ of a sequence $\seq s$, in which the steps are $(i_{1},h_{1}),(i_{2},h_{2}),\ldots,(i_{f},h_{f})$, is \emph{normal} if the following condition holds: For any $1\leq t<f$, if $i_{t+1}<i_{t}$ then $i_{t+1}+2h_{t+1}\ge i_{t}$. \end{defn}
The following claim shows that if we limit ourselves to normal deduplication processes, we can still encode every binary sequence with processes of the same length.
\begin{claim}
\label{clm:normal}For any deduplication process $\seq s=\seq[0]s\reducesto \seq[1]s\reducesto \seq[2]s\reducesto\dotsm\reducesto \seq[f]s=\root(\seq s)$ of length $f$ of a sequence $\seq s$, there is a normal deduplication process $\seq s=\seq[0]s\reducesto \seq[1]s'\reducesto \seq[2]s'\reducesto\dotsm\reducesto \seq[f]s'=\seq[f]s$ of the same length, with the same final sequence.
\end{claim}
\begin{proof} Among all deduplication processes of length $f$ starting with $\seq{s}$ and ending with $\seq[f]{s}$, consider the one minimizing the number of pairs $(i_{t},h_{t})$, $(i_{q},h_{q})$ with $1\leq t<q\leq f$, and $i_{q}<i_{t}$. We claim that this process is normal. Indeed, otherwise there is some $t$, $1\leq t<f$ so that $i_{t+1}<i_{t}$ and $i_{t+1}+2h_{t+1}< i_{t}$. But in this case we can switch the steps $(i_{t},h_{t})$ and $(i_{t+1},h_{t+1})$, performing the step $(i_{t+1},h_{t+1})$ just before $(i_{t},h_{t})$. This will clearly leave all sequences $\seq[0]{s},\seq[1]{s},\ldots,\seq[f]{s}$, besides ${\seq s}_{t}$, the same, and in particular $\seq[0]{s}=\seq{s}$ and $\seq[f]{s}=\root(\seq s)$ stay the same. This contradicts the minimality in the choice of the process, establishing the claim.
\end{proof}
We now turn to the proof of Theorem~\ref{t21}. \begin{proof}[Proof of Theorem~\ref{t21}] Let $U_{\alpha}$ denote the set of $\left\lfloor 2^{n\alpha}\right\rfloor $ sequences that have the smallest duplication distances to their roots among binary sequences of length $n$ and recall that $F_{\alpha}=$ $\max\left\{ f(\seq{s}):\seq{s}\in U_{\alpha}\right\} $. By Claim~\ref{clm:normal}, for each of the sequences $\seq{s}$ of $U_{\alpha}$, there is a normal deduplication process $\seq s=\seq[0]s\reducesto \seq[1]s\reducesto \seq[2]s\reducesto\dotsm\reducesto \seq[f]s$ of length $f\leq F_{\alpha}$. Let the steps of this process be $(i_{1},h_{1}),$ $(i_{2},h_{2}),\ldots,(i_{f},h_{f})$. As before, it is clear that knowing the final sequence $\seq[f]{s}$ and all the pairs $(i_{t},h_{t})$, we can reconstruct $\seq{s}$. There are $6$ possibilities for $\seq[f]{s}$. As each step $(i_{t},h_{t})$ reduces the length of the sequence by $h_{t}$, it follows that $\sum_{i=1}^{f}h_{t}<n$ and therefore there are at most ${{n+f} \choose f}$ possibilities for the sequence $(h_{1},h_{2},h_{3},\dots,h_{f})$. In order to record the sequence $(i_{1},i_{2},\ldots,i_{f})$ it suffices to record $i_{1}$ and all the differences $i_{t}-i_{t+1}$ for all $1\leq t<n$. There are less than $n$ possibilities for $i_{1}$, and there are at most $2^{f}$ possibilities for deciding about the set of all indices $t$ for which the difference $i_{t}-i_{t+1}$ is positive. As the process is normal, for each such positive difference, we know that $i_{t+1}+2h_{t+1}\ge i_{t}$, that is $i_{t}-i_{t+1}\le2h_{t+1}$. It follows that the sum of all positive differences, $\sum_{t:i_{t}-i_{t+1}>0}\left(i_{t}-i_{t+1}\right)$, is at most $2\sum_{t}h_{t}<2n$, and hence the number of choices for these differences is at most ${2n+f \choose f}$.

Since $i_{f}\le3$, we have $i_{1}-i_{f}\ge1-3=-2$. So
\begin{align*}
\phantom{xx}\sum_{\mathclap{t:i_{t}-i_{t+1}\le0}}\left(i_{t}-i_{t+1}\right)
=
(i_1-i_f)-\sum_{\mathclap{t:i_{t}-i_{t+1}>0}}\left(i_{t}-i_{t+1}\right)
>-2-2n.
\end{align*}
Therefore, the number of choices for all non-positive differences $i_{t}-i_{t+1}$ is at most ${2n+f+2 \choose f}$. Putting all of these together, and noting that $\left|U_{\alpha}\right|\ge2^{n\alpha}-1$, implies the assertion of Theorem~\ref{t21}.
\end{proof}
Since $\binom{p}{q}\le2^{pH(q/p)}$ for positive integers $0<q<p$~\cite[p.~309]{macwilliams1977theory}, Theorem~\ref{t21} implies that 
$$3\left(2+\frac{F_{\alpha}}{n}\right)H\left(\frac{F_{\alpha}/n}{2+F_{\alpha}/n}\right)+\frac{F_{\alpha}}{n}\ge\alpha+o(1),$$
where $H$ is the binary entropy function, $H(x) = -x \log_2 x -(1-x) \log_2(1-x)$. The expression on the left side of the inequality is strictly increasing in $\frac{F_{\alpha}}{n}$, and it is less than $0.99$ if we substitute $\frac{F_{\alpha}}{n}$ by $0.045$. If we let $\alpha=0.99$, it follows that for sufficiently large $n$, we have $\frac{F_{\alpha}}{n}\ge0.045$, thereby establishing the lower bound in Theorem~\ref{thm:bounds}.

To prove the upper bound in Theorem~\ref{thm:bounds}, we prove the following theorem.

\begin{thm}
\label{thm:ub}The limit $\lim_{n\to\infty}f(n)/n$ exists and for all $n$, $f(n)\le\frac{2}{5}n+15$.\end{thm}
\begin{proof}
Note that for any positive integers $n$ and $m$, $f(n+m)\leq f(n)+f(m)+2$. Indeed, given a sequences of length $n+m$ we can deduplicate separately its first $n$ bits and its last $m$ bits, getting a concatenation of two square-free sequences (of total length at most $6$). It then suffices to check that each such concatenation can be deduplicated to its root through at most $2$ additional deduplication steps. Therefore, the function $g(n)=f(n)+2$ is subadditive: 
\[
g(n+m)=f(n+m)+2\leq f(n)+f(m)+4=g(n)+g(m).
\]
Now, by Fekete's Lemma~\cite{steel1997probability}, $g(n)/n$ tends to a limit (which is the infimum over $n$ of $g(n)/n$), and it is clear that the limit of $f(n)/n$ is the same as that of $g(n)/n$. We term this limit the \emph{binary duplicatoin constant}.

This proof of the existence of $\lim_{n\to\infty}f(n)/n$ provides a simple way to derive an upper bound for the limit by computing $f(n)$ precisely for some small $n$. In particular, from Table~\ref{tab:f(n)}, we find $\lim_{n\to\infty}f(n)/n\le(f(32)+2)/32=17/32$. We can improve upon this result as follows.

For positive integers $n,m$, let $f(n,m)$ be the smallest number $k$ such that every sequence of length $n$ can be converted to a sequences of length at most $m$ via $k$ deduplication steps. A sequence of length $n$ can be converted to its root by first repeatedly converting its $a$-substrings to substrings of length at most $b$ via $f(a,b)$ deduplication steps. Thus for integers $a>b>0$, we have
\begin{equation}
f(n)\le\frac{f(a,b)}{a-b}n+\max_{i<a}f(i)\label{eq:fnm}
\end{equation}
With the help of a computer we find the values of $f(n,m)$ for $3\le m<n\le32$. An illustration is given in Figure~\ref{fig:bestRatio}. In particular we have $\frac{f\left(32,12\right)}{20}=\frac{8}{20}=\frac{2}{5}$ from Figure~\ref{fig:bestRatio} and $\max_{i<32}f(i)=15$ from Table~\ref{tab:f(n)}, implying $f(n)\le\frac{2}{5}n+15$.
\end{proof}
\begin{figure}
\centering{}\includegraphics[angle=90,width=0.7\textwidth]{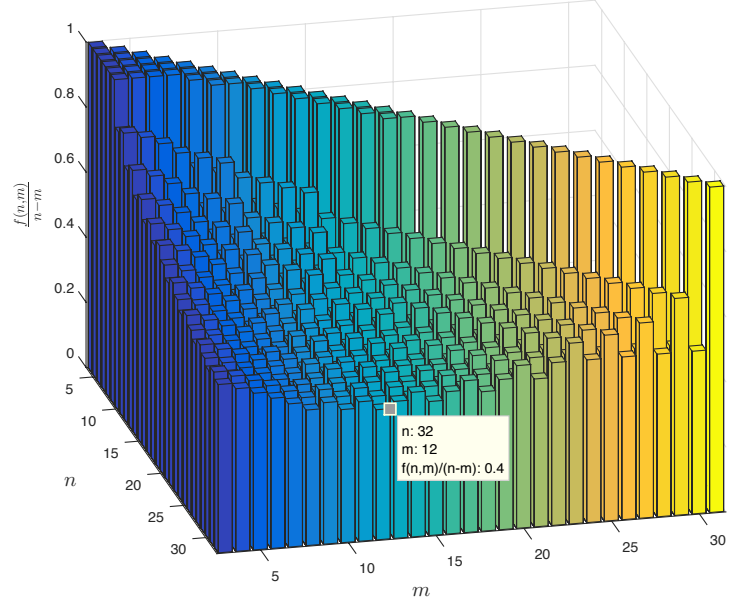}\caption{$\frac{f(n,m)}{n-m}$ for $3\le m<n\le32$.}
\label{fig:bestRatio}
\end{figure}

Weaker upper bounds on $f(n)$ can be obtained without resorting to computation in the following ways. First, to deduplicate a sequence to its root, we first can deduplicate each block of $t$ consecutive identical bits to a single bit by $\lceil\log_{2}t\rceil$ deduplications and then finish in less than $\log_{2}n$ additional steps. This shows that for large $n$ , $f(n)\leq\frac{2}{3}n+o(n)$ (the extremal case for this argument is the one in which each block is of size $3$). Second, it is known that every binary sequence of length at least 19 contains a repeat of length at least 2 \cite{entringer1974nonrepetitive}, implying that $f(n)\leq\frac{1}{2}n+o(n)$.

\paragraph*{Parallel duplication }
One can also define the parallel duplication distance to the root by allowing non-overlapping duplications to occur simultanously, with $f'(n)$ being the maxmimum parallel duplication distance to the root of a sequence of length $n$. Similar to the normal duplication distance it is helpful to think in terms of deduplications. Since each parallel deduplication step decreases the length of a sequence by at most a factor of $2$, $f'(n)>\log_{2}n-2$ (and in fact $f'(\seq{s})\geq\log_{2}n-2$ for every sequence of length $n$.) It is not difficult to see that $f'(n)<2\log_{2}n$ by first deduplicating, in parallel, all blocks of identical elements in the sequence to blocks of size $1$, and then by deduplicating the resulting alternating sequence to its root.

\paragraph*{Partial deduplication}
The definition of $f(n,m)$ gives rise to the following question: For a fixed $0<\alpha\le1$, what is $\lim_{n}\frac{f\left(n,\left\lfloor \alpha n\right\rfloor \right)}{1-\alpha}$, if it exists? At first glance, one may expect $\lim_{n}\frac{f\left(n,\left\lfloor \alpha n\right\rfloor \right)}{1-\alpha}$ to be decreasing in $\alpha$ since if $\alpha$ is large, one may think it is easier to find enough long repeats to reduce the length of the sequence quickly by a factor of $1-\alpha$. However, we show that $\lim_{n}\frac{f\left(n,\left\lfloor \alpha n\right\rfloor \right)}{n\left(1-\alpha\right)}=\lim_{n}\frac{f(n)}{n}$. 

Let $\gamma=\lim_n\frac{f(n)}{n}$. For $\epsilon>0$, there exists $k$ such that for all $n>k$, $f(n)\le(\gamma+\epsilon)n$. Thus
\begin{equation}\label{eq:f1}
f(n,\lfloor\alpha n\rfloor)\le f\left(n-\lfloor\alpha n\rfloor+3\right)\le(\gamma+\epsilon)\left({\left(1-\alpha\right)n}+4\right).
\end{equation}
On the other hand, let $\delta = \liminf_n \frac{f(n,\lfloor\alpha n\rfloor)}{\left(1-\alpha\right)n}$. For $\epsilon>0$, there exists $k$ such $f(k,\lfloor\alpha k\rfloor)\le (\delta+\epsilon)(1-\alpha)k$. Hence,
\begin{equation}\label{eq:f2}
f(n)\le\frac{f(k,\lfloor\alpha k\rfloor)}{k-\lfloor\alpha k\rfloor}n+k\le (\delta+\epsilon)n+k.
\end{equation}
The result follows by dividing~\eqref{eq:f1} by $(1-\alpha)n$ and taking a $\limsup_n$ and by dividing~\eqref{eq:f2} by $n$ and taking a $\lim_n$.



\section{Duplication Distance for L-systems\label{sec:Lsystems}}
\emph{L-systems}, or Lindenmayer systems are sequence rewriting systems developed by Lindemayer in 1968 \cite{lsystems}. He used them in the context of biology to model the growth process of plant development. He introduced context-free as well as context-sensitive L-systems. Here we will discuss distance to the root for sequences arising in context-free L-systems, also known as 0L-systems. A 0L-system comprises three components: 
\begin{itemize}
\item Alphabet ($\Sigma$): An alphabet of symbols used to construct sequences. 
\item Axiom sequence or initiator ($\seq\omega$): The starting sequence from which a 0L-system is constructed. 
\item Production rule ($h$): A rule that constructs new sequences by expanding each symbol in a given sequence into a sequence of symbols. The production rule is represented by the function $h: \Sigma^* \rightarrow \Sigma^*$, which for any two sequences  $\seq a$ and $\seq b \in \Sigma^*$ satisfies
$$h(\seq a\seq b) = h(\seq a)h(\seq b)$$
 where $h(\seq a)h(\seq b)$ represents the concatenation of $h(\seq a)$ and $h(\seq b)$. The production rule $h$ can be deterministic or stochastic. Here we consider only deterministic rules.  Such 0L-systems with deterministic $h$ are denoted as D0L-systems~\cite{DOL}.
\end{itemize}
\begin{exm}[\emph{Fibonacci words}]
Consider $\Sigma = \{X,Y\}$, $\seq\omega = X$, and \[h(X) = XY, \quad h(Y) = X.\] For this D0L-system, the first $5$ sequences are as follows:
\begin{align*}
h^0(\seq\omega) &= X\\
h^1(\seq\omega) &=  XY\\
h^2(\seq\omega) &= XYX\\
h^3(\seq\omega) &= XYXXY\\
h^4(\seq\omega) &= XYXXYXYX\\
h^5(\seq\omega) &= XYXXYXYXXYXXY
\end{align*}
This can also be represented by the following tree:

\Tree[.X [.X [.X [.X [.X ]
               [.Y ]]
          [.Y 
                [.X  
                          ]]]
                [.Y [.X [.X ]
               [.Y ]]
          ]]
          [.Y [.X [.X [.X ]
               [.Y ]]
          [.Y [.X ]
               ]]
                ]]      
          
These sequences are called Fibonacci words as they satisfy $$h^n(\seq\omega) = h^{n-1}(\seq\omega)h^{n-2}(\seq\omega)~ \forall ~ n\geq 2.$$
\end{exm}

\begin{exm}[\emph{Thue-Morse Sequence}] 
Let $\Sigma = \{0,1\}$, $\seq\omega = 0$, and \[ h(0) = 01,\quad h(1) = 10.\] For this D0L-system the tree of sequence generation is given below:

\Tree[.0 [.0 [.0 [.0 [.0 ]
               [.1 ]]
          [.1 [.1 ]
                [.0  
                          ]]]
                [.1 [.1 [.1 ]
               [.0 ]]
          [.0 [.0 ]
                [.1  
                          ]]]]
          [.1 [.1 [.1 [.1 ]
               [.0 ]]
          [.0 [.0 ]
                [.1  
                          ]]]
                [.0 [.0 [.0 ]
               [.1 ]]
          [.1 [.1 ]
                [.0  
                          ]]]]]

The sequence generated by this D0L-system are called Thue-Morse sequences. Alternatively, the Thue-Morse sequences can be defined recursively by starting with $\seq[0]t=0$ and forming $\seq[i+1]t$ by concatenating $\seq[i]t$ and its complement $\overline{\seq[i]t}$. 
\end{exm}

 We show that binary D0L-systems, which have production rules of the form $h(0)=\seq{u}$ and $h(1)=\seq{v}$, with $\seq{u},\seq{v}\in\left\{ 0,1\right\} ^{*}$ have a logarithmic distance to their roots. 
\begin{restatable}{lem}{lemdzerol}\label{lem:d0L}
For any binary D0L-system with initiator $\seq\omega $
and production rule $h$, we have 
\[f\left(h^{r}(\seq\omega )\right)=\Theta\left(\log_2\left|h^{r}(\seq\omega )\right|\right),\qquad\text{as }r\to\infty.\]\end{restatable}

\begin{proof}
For any sequence $\seq t$, since $f(\seq t)\ge\log_2|\seq t|$, we have $f\left(h^{r}(\seq \omega)\right)\ge\log_2\left|h^{r}(\seq \omega)\right|$. It remains to show that $f\left(h^{r}(\seq \omega)\right)=O\left(\log_2\left|h^{r}(\seq \omega)\right|\right)$. We start by proving the following claim.

\begin{claim*}For any binary D0L-system with initiator $\omega$ and production rule $h$, we have
\begin{equation}
f\left(h^{r}(\seq \omega)\right)\le f\left(h^{r-1}(\seq \omega)\right)+c\le f(\seq \omega)+rc,\label{eq:L-sys}
\end{equation}
where $c=\max_{\seq{z}\in\left\{ 0,1,01,10,010,101\right\} }f\left(h(\seq{z})\right)$. 
\end{claim*}

To prove the claim, let $\seq{x}=h^{r-1}(\seq{\omega})$ and $\seq y=h^{r}(\seq{\omega})$ and consider the sequence of deduplications that turns $x$ into its root $\seq{z}\in\left\{ 0,1,01,10,010,101\right\} $. We can deduplicate $\seq{y}$ in a similar manner to $h(\seq{z})$: For each step in the deduplication process of $\seq{x}$ that deduplicates a substring $a_{1}\dotsm a_{k}a_{1}\dotsm a_{k}$ to $a_{1}\dotsm a_{k}$, we deduplicate $h\left(a_{1}\right)\dotsm h\left(a_{k}\right)h\left(a_{1}\right)\dotsm h\left(a_{k}\right)$ to $h\left(a_{1}\right)\dotsm h\left(a_{k}\right)$ in the deduplication process of $\seq{y}$, resulting eventually in $h(\seq{z})$. This completes the proof of the claim.

We now turn to proving $f\left(h^{r}(\seq{\omega})\right)=O\left(\log_2\left|h^{r}(\seq{\omega})\right|\right)$. If $\left|h^{r}(\seq{\omega})\right|=O(1)$, then $f\left(h^{r}(\seq{\omega})\right)=O(1)$ as well, and there is nothing to prove. If $\left|h^{r}(\seq{\omega})\right|=2^{\Omega(r)}$, then $r=O\left(\log_2\left|h^{r}(\seq{\omega})\right|\right)$ and the desired result follows from~(\ref{eq:L-sys}). The last case that we need to consider is when $\left|h^{r}(\seq{\omega})\right|\to\infty$ but $\left|h^{r}(\seq{\omega})\right|=2^{o(r)}$. Without loss of generality, assume $\left|h(1)\right|\ge\left|h(0)\right|$. Then the condition $\left|h^{r}(\seq{\omega})\right|=2^{o(r)}$ can be shown to occur only if the initiator $\seq{\omega}$ contains at least one occurrence of $1$, $h(0)=0$, and $h(1)$ has exactly one occurrence of $1$ and one or more 0s. In this case, the number of 1s in $h^{r}(\seq{\omega})$ is constant and again $f\left(h^{r}(\seq{\omega})\right)=O\left(\log_2\left|h^{r}(\seq{\omega})\right|\right)$.
\end{proof}

The previous lemma shows that the duplication distances to the root for both of Fibonacci words and Thue-Morse sequences are logarithmic in sequence length. This is particularly interesting in the case of the Thue-Morse sequence. Despite the fact that the Thue-Morse sequence grows by taking the complement, it contains enough repeats to allow a logarithmic distance. Note also that the Thue-Morse sequence is used to generate ternary square-free sequences. 

In the next lemma, we give better bounds than those that can be obtained from Lemma~\ref{lem:d0L} or~(\ref{eq:L-sys}) for Thue-Morse and Fibonacci  sequences.
\begin{lem}\label{lem:TMF}
Let $\seq[r]t$ and $\seq[r]u$ denote the $r$th Thue-Morse and Fibonacci words, respectively. For $r\ge2$, we have
\begin{align*}
f\left(\seq[r]t\right) & \le2r,\\
f\left(\seq[r]u\right) & \le r.
\end{align*}
\end{lem}

\begin{proof}
We first prove the upper bound for ${\seq t}_{r}$. For $r\ge3$, we have 
\begin{align*}
f\left({\seq t}_{r}\right) & =f\left({\seq t}_{r-1}\overline{{\seq t}}_{r-1}\right)\\
 & =f\left({\seq t}_{r-2}\overline{{\seq t}}_{r-2}\overline{{\seq t}}_{r-2}{\seq t}_{r-2}\right)\\
 & \le1+f\left({\seq t}_{r-2}\overline{{\seq t}}_{r-2}{\seq t}_{r-2}\right)\\
 & =1+f\left({\seq t}_{r-3}\overline{{\seq t}}_{r-3}\overline{{\seq t}}_{r-3}{\seq t}_{r-3}{\seq t}_{r-3}\overline{{\seq t}}_{r-3}\right)\\
 & \le3+f\left({\seq t}_{r-3}\overline{{\seq t}}_{r-3}{\seq t}_{r-3}\overline{{\seq t}}_{r-3}\right)\\
 & \le4+f\left({\seq t}_{r-3}\overline{{\seq t}}_{r-3}\right)\\
 & =4+f\left({\seq t}_{r-2}\right).
\end{align*}
If $r\ge3$ is even, then $f\left({\seq t}_{r}\right)\le4\frac{r-2}{2}+f\left({\seq t}_{2}\right)=2\left(r-2\right)+1=2r-3$; and if $r\ge3$ is odd, then $f\left({\seq t}_{r}\right)\le4\frac{r-1}{2}+f\left({\seq t}_{1}\right)=2\left(r-1\right)$. This completes the proof of the first claim.

We now turn to $f\left({\seq u}_{r}\right)$. The $r$th Fibonacci word can be obtained via the following recursion: ${\seq u}_{r}={\seq u}_{r-1}{\seq u}_{r-2}$ for $r\ge2$ and ${\seq u}_{0}=0$, ${\seq u}_{1}=01$. If $r\ge5$, then
\[
\begin{split}{\seq u}_{r} & ={\seq u}_{r-1}{\seq u}_{r-2}\\
 & ={\seq u}_{r-2}{\seq u}_{r-3}{\seq u}_{r-3}{\seq u}_{r-4}\\
 & ={\seq u}_{r-2}{\seq u}_{r-3}{\seq u}_{r-4}{\seq u}_{r-5}{\seq u}_{r-4}\\
 & ={\seq u}_{r-2}^{2}{\seq u}_{r-5}{\seq u}_{r-4}.
\end{split}
\]
Hence, $f\left({\seq u}_{r}\right)\le1+f\left({\seq u}_{r-2}{\seq u}_{r-5}{\seq u}_{r-4}\right)$. Noting that ${\seq u}_{r-2}{\seq u}_{r-5}{\seq u}_{r-4}={\seq u}_{r-3}{\seq u}_{r-4}{\seq u}_{r-5}{\seq u}_{r-4}={\seq u}_{r-3}^{2}{\seq u}_{r-4}$, we write
\begin{align*}
f\left({\seq u}_{r}\right) & \le1+f\left({\seq u}_{r-2}{\seq u}_{r-5}{\seq u}_{r-4}\right)\\
 & =1+f\left({\seq u}_{r-3}^{2}{\seq u}_{r-4}\right)\\
 & \le2+f\left({\seq u}_{r-3}{\seq u}_{r-4}\right)\\
 & =2+f\left({\seq u}_{r-2}\right).
\end{align*}
Now, if $r\ge5$ is even, then $f\left({\seq u}_{r}\right)\le\left(r-4\right)+f\left({\seq u}_{4}\right)\le r-2$ since $f\left({\seq u}_{4}\right)=f\left(01001010\right)\le2$; and if $r\ge5$ is odd, then $f\left({\seq u}_{r}\right)\le\left(r-3\right)+f\left({\seq u}_{3}\right)\le r-1$ as $f\left({\seq u}_{3}\right)=f\left(01001\right)\le2$.
\end{proof}

\section{Approximate-duplication distance}\label{sec:mismatch}
Recall that $f_{\beta}(n)$ is the least $k$ such that every sequence of length $n$ can be converted to a square-free sequence in $k$ approxmiate deduplication steps, with at most a $\beta$ fraction of mismatches in each step. In this section,  we provide bounds on $f_{\beta}(n)$ for $\beta<1/2$ and $\beta>1/2$. We first however present some useful definitions. 

For $0\le\beta<1$, a \emph{$\beta$-repeat of length} $h$ in a binary sequence consists of two consecutive blocks in the sequence, each of length $h$, such that the Hamming distance between them is at most $\beta h$. If $\seq{u}\seq{v}\seq{v}'\seq{w}$ is a binary sequence, and $\seq{v}\seq{v}'$ is a $\beta$-repeat, then a $\beta$\emph{-deduplication} produces $\seq{u}\seq{v}\seq{w}$ or $\seq{u}\seq{v}'\seq{w}$. Note that in this case the set of roots of $\seq s$ is not necessarily unique, but the length of any root is at most 3, even if $\beta=0.$ 

The next theorem establishes a sharp phase transition in the behavior of $f_{\beta}(n)$ at $\beta=1/2$. Its proof relies on Theorem~\ref{thm:rep_beta}, which guarantees the existence of $\beta$-repeats under certain conditions. In what follows, for an integer $m$, we use $[m]$ to denote $\{1,\dotsc,m\}$.

\thmbetalessthanhalf* 
\begin{proof}
	The proof for $\beta<1/2$ is similar to the proof of the lower bound in Theorem~\ref{thm:bounds}. In this case however, to make the deduplication process reversible, for every deduplication we need to record whether it is of the form $\seq u \seq v\seq v'\seq w\reducesto \seq u \seq v\seq w$ or of the form $\seq u \seq v'\seq v\seq w\reducesto \seq u \seq v\seq w$, and we must also encode the sequence $\seq v'$. In the $t$th deduplication step, we have $|\seq v|=|\seq v'|=h_{t}$. Since $\seq v'$ is in the Hamming sphere of radius $\beta h_{t}$ around $\seq v$, there are at most $2^{h_{t}H(\beta)}$ options for $\seq v'$~\cite[Lemma 4.7]{roth2006introduction}. Thus
\[
6n\sum_{f=1}^{F_{\beta}}\binom{n+f}{f}\binom{2n+f}{f}\binom{2n+f+2}{f}2^{nH(\beta)}2^{2f}\ge2^{n},
\]
where $F_{\beta}=f_{\beta}(n)$ and we have used $\sum_{t}h_t\le n$. The desired result then follows since $H(\beta)<1$.

Suppose $\beta>1/2$. Let $K=\left\lceil \frac{2\beta+1}{2\beta-1}\right\rceil ^{2}$ and $\epsilon=C-K$. Note that $\epsilon>0$. By appropriately choosing $C_{1}$, we can have $f_{\beta}(i)\le\left(K+\frac{\epsilon}{2}\right)\ln i+C_{1}$ for all $i<M$, where $M$ is sufficiently large and in particular $M>K$. Assuming that this holds also for all $i<n$, where $n\ge M$, we show that it holds for $i=n$. From Theorem~\ref{thm:rep_beta}, every binary sequence $\seq s$ of length $n$ has a $\beta$-repeat of length $\ell\lfloor n/K\rfloor$ for some $\ell\in\left[\sqrt{K}\right]$, implying 
\begin{align*}
f_{\beta} (\seq s)
 &\le
f_{\beta}\left(n-\ell\left\lfloor\frac{n}{K}\right\rfloor\right)+1\\
& \le\left(K+\frac{\epsilon}{2}\right)\ln\left(n-\left\lfloor\frac{n}{K}\right\rfloor\right)+1+C_{1}\\
 & \le\left(K+\frac{\epsilon}{2}\right)\ln n-\frac{\left(K+\frac{\epsilon}{2}\right)\left(n-K\right)}{Kn}+1+C_{1}\\
 & \le\left(K+\frac{\epsilon}{2}\right)\ln n+C_{1}\\
&\le C\ln n,
\end{align*}
where the last two steps hold for sufficiently large $n$. Hence, $f_\beta(n)\le C\ln n$.
\end{proof}
\begin{thm}
\label{thm:rep_beta}If $\beta>\frac{1}{2}$, then for any integer $k\ge\frac{2\beta+1}{2\beta-1}$, any binary sequence of length $n$ contains a $\beta$-repeat of length $\ell\lfloor n/k^{2}\rfloor$ for some $\ell\in[k]$.\end{thm}
\begin{proof}
	Let $k$ be a positive integer to be determined later and put $K=k^{2}$. Furthermore, let ${\seq s}'={\seq s}_{1}\dotsm {\seq s}_{K}$ be a partition of the first $KB$ symbols of $\seq s$ into blocks of length $B=\left\lfloor\frac{n}{K}\right\rfloor$. We now consider as a code~\cite{macwilliams1977theory} the $k+1$ binary vectors
\[
{\seq t}_{i}={\seq s}_{i}\dotsm {\seq s}_{i+K-k-1},\qquad\left(1\le i\le k+1\right),
\]
each of length $m=\left(K-k\right)B$. By Plotkin's bound~\cite[p.~41]{macwilliams1977theory}, the minimum Hamming distance of this code is at most $\left(\frac{1}{2}+\frac{1}{2k}\right)m$. Thus there exist ${\seq t}_{i}$ and ${\seq t}_{j}$ with $i<j$ with Hamming distance at most $\left(\frac{1}{2}+\frac{1}{2k}\right)m$.

Put $h=\left(j-i\right)B$ and let $m'=h\lfloor m/h\rfloor$ be the largest integer which is at most $m$ and is divisible by $h$. Let ${\seq t}'_{i}$ and ${\seq t}'_{j}$ consist of the first $m'$ bits of ${\seq t}_{i}$ and ${\seq t}_{j}$, respectively. The Hamming distance between ${\seq t}'_{i}$ and ${\seq t}'_{j}$ is clearly still at most $\left(\frac{1}{2}+\frac{1}{2k}\right)m$. But $\left(\frac{1}{2}+\frac{1}{2k}\right)m\le\left(\frac{1}{2}+\frac{1}{k-1}\right)m'$ since
\begin{align*}
\left(\frac{1}{2}+\frac{1}{2k}\right)m & =\left(\frac{1}{2}+\frac{1}{2k}\right)\frac{m}{m'}m'
\stackrel{(*)}{\le}
\left(\frac{1}{2}+\frac{1}{2k}\right)\frac{k}{k-1}m'
=
\left(\frac{1}{2}+\frac{1}{k-1}\right)m',
\end{align*}
where $(*)$ can be proved as follows. By the definition of $m'$, we have $m-m'<h$. Additionally, $h\le k B$ since $1\le i< j\le k+1$. So,  
\[
\frac{m-m'}B< k,
\] which since $B$ divides $m,m'$, implies $\frac{m-m'}B\le k-1$ and, in turn, $m'\ge m-\left(k-1\right)B=\left(k-1\right)^{2}B$. Hence $\frac m{m'} \le \frac{k(k-1)B}{(k-1)^2B}=\frac k{k-1}$.

Split ${\seq t}'_{i}$ and ${\seq t}'_{j}$ into blocks of length $h$ each: ${\seq t}'_{i}={\seq z}_{1}{\seq z}_{2}\cdots {\seq z}_{p}$, ${\seq t}'_{j}={\seq z}_{2}{\seq z}_{3}\cdots {\seq z}_{p}{\seq z}_{p+1}$, where $p=m'/h$. The Hamming distance between ${\seq t}'_{i}$ and ${\seq t}'_{j}$ is the sum of the Hamming distances between ${\seq z}_{q}$ and ${\seq z}_{q+1}$ as $q$ ranges from $1$ to $p$. Thus, by averaging, there exists an index $r$ so that the Hamming distance between ${\seq z}_{r}$ and ${\seq z}_{r+1}$ is at most $\left(\frac{1}{2}+\frac{1}{k-1}\right)h$. Putting $k\ge\frac{2\beta+1}{2\beta-1}$ so that $\frac{1}{2}+\frac{1}{k-1}\le\beta$ ensures that ${\seq z}_{r}{\seq z}_{r+1}$ is $\beta$-repeat of length $h=\left(j-i\right)B=\left(j-i\right)$$\lfloor n/K\rfloor$. \end{proof}

Let a $\beta_{h}$-repeat be a repeat of length $h$ with at most $h\beta_{h}$ mismatches, i.e., the two blocks are at Hamming distance at most $h\beta_h$. In the preceding theorems and their proofs, in principal, we do not need the maximum number of permitted mismatches to be a linear function of the length of the repeat, so we can apply the same techniques to $\beta_h$-repeats with nonlinear relationships:
\begin{restatable}{thm}{thmbetaabouthalf}
Let $\beta_{h}^{a}=\frac{1}{2}+\frac{1}{h^{a}}$, where $0<a<1$ is a constant, and let $f_{a}(n)$ be the smallest number $f$ such that any binary sequence of length $n$ can be deduplicated to a root in $f$ steps by deduplicating $\beta_{h}^{a}$-repeats. There exist positive constants $c_{2},c_{3}$ such that
\begin{equation}
f_{a}(n)\le c_{2}n^{2a/(1+a)}+c_{3}.\label{eq:f_a}
\end{equation}
\end{restatable}
\begin{proof} By making appropriate changes to the proof of Theorem~\ref{thm:rep_beta}, one can show that for $k=\left\lceil 2n^{a/(1+a)}\right\rceil$, every binary sequence of sufficiently long length $n$ contains a $\beta_{h}^{a}$-repeat of length $h=\ell\lfloor n/k^{2}\rfloor$, for some $\ell\in[k]$. To do so, we need to prove $\left(\frac{1}{2}+\frac{1}{k-1}\right)h\le\beta_h^ah$ for all $h$ of the form $h=\ell \lfloor n/k^2\rfloor$, $\ell\in[k]$. This holds since with the aforementioned value of $k$,
\begin{equation*}
\beta_{\ell\lfloor n/k^{2}\rfloor}^{a}=\frac{1}{2}+\frac{1}{\left(\ell\lfloor n/k^{2}\rfloor\right)^{a}}
\ge\frac{1}{2}+\frac{1}{\left(k\lfloor n/k^{2}\rfloor\right)^{a}}
\ge\frac{1}{2}+\frac{1}{k-1},
\end{equation*}
for all $\ell\in[k]$ and sufficiently large $n$. 

We can now prove~(\ref{eq:f_a}) by induction. Clearly, for any $M$, there exist constants $c_{2},c_{3}$ such that $f_{a}(i)\le c_{2}i^{2a/(1+a)}+c_{3}$ for all $i\le M$. Choose $M$ to be sufficiently large as to satisfy the requirements of the rest of the proof. Fix $n>M$ and assume that $f_{a}(i)\le c_{2}i^{2a/(1+a)}+c_{3}$ for all $i<n$. Since in every sequence of length $n$, there exists a $\beta_{h}^{a}$-repeat with $h=\ell\lfloor n/k^{2}\rfloor$, for some $\ell\in[k]$ and $k=\left\lceil 2n^{a/(1+a)}\right\rceil $, it holds that
\begin{align*}
f_{a}  (n)&\le1+c_{2}\left(n-\ell\lfloor n/k^{2}\rfloor\right)^{2a/(1+a)}+c_{3}\\
 & \le1+c_{2}\left(n-\frac{1}{5}n^{\frac{1-a}{1+a}}\right)^{2a/(1+a)}+c_{3}\\
 & =1+c_{2}n^{2a/(1+a)}\left(1-\frac{1}{5}n^{-\frac{2a}{1+a}}\right)^{2a/(1+a)}+c_{3}\\
 & \le1+c_{2}n^{2a/(1+a)}\left(1-\frac{2a}{5\left(1+a\right)}n^{-\frac{2a}{1+a}}\right)+c_{3}\\
 & =c_{2}n^{2a/(1+a)}+\left(1-\frac{2ac_{2}}{5\left(1+a\right)}\right)+c_{3}\\
 & \le c_{2}n^{2a/(1+a)}+c_{3},
\end{align*}
where the inequalities hold for sufficiently large $n$. The third inequality follows from Bernoulli's inequality and the the last one follows from the fact that we can choose $c_{2}$ to be arbitrarily large. \end{proof}

\section{Duplication distances for different roots\label{sec:stringSystems}}

In this section, we study $f_{\seq{\sigma}}$ for $\seq{\sigma}\in\{0,1,01,10,010,101\}$. 
It is easy to see that $f_{0}(n)=f_{1}(n)=\left\lceil \log_2 n\right\rceil .$ Clearly $f_{10}=f_{01}$ and $f_{101}=f_{010}$. So we limit our attention to roots $\seq \sigma=10$ and $\seq \sigma=101$. Plots for $f_{10}(n)$ and $f_{101}(n)$, obtained through computer search, are given in Figure~\ref{fig:f10-101(n)}.

\begin{figure}
{\begin{centering}
\includegraphics[width=3in]{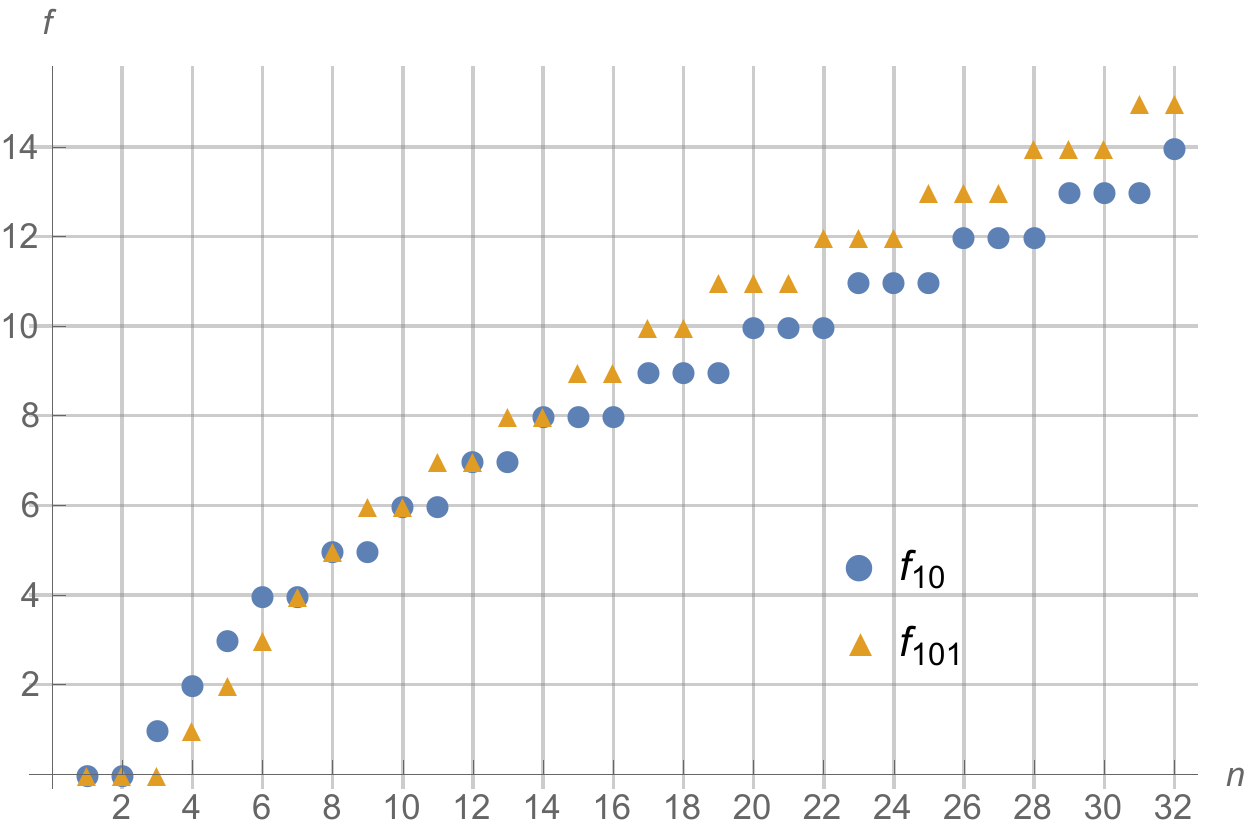}
\par\end{centering}
}
\caption{$f_{10}(n)$ and $f_{101}(n)$ for $1\le n\le32$.\label{fig:f10-101(n)}}
\end{figure}

\thmdiffrentroots*
\begin{proof}
The general approach in this proof is similar to that of the proof of Fekete's lemma in \cite{steel1997probability}. We prove the theorem for $\lim_{n}\frac{f_{10}(n)}{n}$. The proof for $\frac{f_{101}(n)}{n}$ is similar.

Let $\gamma=\liminf_{n}\frac{f_{10}(n)}{n}$ and let $k\ge3$ be such that $f_{10}(k)+5+2\log_2 k\le k\left(\gamma+\epsilon\right)$ for $\epsilon>0$. Let $\seq s$ be a sequence of length $n$. Starting from the beginning of $\seq s$, partition it into substrings that are the shortest possible while having length at least $k$ and different symbols at the beginning and the end (so that their root is either 10 or 01). Name these substrings ${\seq s}_{1},\dotsc,{\seq s}_{m+1}$, where $\left|{\seq s}_{i}\right|\ge k$ for $i\le m$ and $1\le\left|{\seq s}_{m+1}\right|\le k$. Let $s_{i,j}$ denote the $j$th element of ${\seq s}_{i}$. We deduplicate ${\seq s}$ to its root by first deduplicating its substrings ${\seq s}_{i}$ to their roots.

For each substring ${\seq s}_{i}$ of the partition, except the last one, we consider the following cases and deduplicate ${\seq s}_{i}$ as indicated, where without loss of generality we assume ${\seq s}_{i}$ starts with 1 and ends with 0:
\begin{itemize}
\item $\left|{\seq s}_{i}\right|=k$: Deduplicate this substring to 10 in $f_{10}(k)$ steps.
\item $\left|{\seq s}_{i}\right|>k$ and $s_{i,k-1}=1$: In this case, ${\seq s}_{i}=1\seq x11,\negthinspace1^{*}0$, where $\seq x\in\left\{ 0,1\right\} ^{k-3}$, for clarity a comma is placed after the $k$th element of ${\seq s}_{i}$, and $a^*$ denotes that the symbol $a$ appears 0 or more times. We reduce the length of the last run of 1s in ${\seq s}_{i}$ by $\left|{\seq s}_{i}\right|-k$ in $\left\lceil \log_2\left(\left|{\seq s}_{i}\right|-k+1\right)\right\rceil $ deduplication steps to obtain $1\seq x10$. Then deduplicate the result to 10 in $f_{10}(k)$ steps.
\item $\left|{\seq s}_{i}\right|>k$ and $s_{i,k-1}=0$: In this case, ${\seq s}_{i}=1\seq x01,\negthinspace1^{*}0$, where $\seq x\in\left\{ 0,1\right\} ^{k-3}$ and where a comma is placed after the $k$th element of ${\seq s}_{i}$. We reduce the length of the last run of 1s in ${\seq s}_{i}$ by $\left|{\seq s}_{i}\right|-k-1$ in $\left\lceil \log_2\left(\left|{\seq s}_{i}\right|-k\right)\right\rceil $ deduplication steps to obtain $\hat{\seq s}_{i}=1\seq x01,0$ and note that $\hat{\seq s}_{i}$ has length $k+1$ and ends with $010$. Now either $\hat{\seq s}_{i}$ has a run of length at least 2 or not. If it does, we reduce the length of this run by 1 to obtain a sequence of length $k$, which we then convert to $10$ in $f_{10}(k)$ deduplication steps. If not, then $\hat{\seq s}_{i}$ is an alternating sequence of the form $101010\dotsm10$ which can be deduplicated to $10$ in no more than $\left\lceil \log_2\frac{k+1}{2}\right\rceil $ steps.
\end{itemize}
The resulting sequences has length at most $2m+k$ and can be deduplicated to its root in at most as many steps. We thus have
\begin{align*}
f(n) & \le mf_{10}(k)+\sum_{i=1}^{m}\left\lceil \log_2\left(\left|{\seq s}_{i}\right|-k+1\right)\right\rceil +m\left\lceil \log_2\frac{k+1}{2}\right\rceil +3m+k\\
 & \le mf_{10}(k)+\sum_{i=1}^{m}\log_2\left|{\seq s}_{i}\right|+m\log_2 k+5m+k\\
 & \le\frac{n}{k}f_{10}(k)+\frac{2n}{k}\log_2 k+5\frac{n}{k}+k,
\end{align*}
 where for the last step we have used the fact that 
\[
\sum_{i=1}^{m}\log_2\left|{\seq s}_{i}\right|\le m\log_2\left(n/m\right)\le\frac{n}{k}\log_2 k
\]
which holds since $\sum_{i=1}^{m}\left|{\seq s}_{i}\right|\le n$, $\frac{d}{dm}m\log_2\frac{n}{m}>0$ and $m\le\frac{n}{k}$. It follows that
\[
\frac{f(n)}{n}\le\frac{f_{10}(k)}{k}+\frac{2\log_2 k}{k}+\frac{5}{k}+\frac{k}{n}\le\gamma+\epsilon+\frac{k}{n}\ .
\]
Taking $\lim$ of both sides and noting that $\epsilon>0$ is arbitrary proves that $\lim_n \frac {f(n)}{n}\le\liminf_n\frac{f_{10}(n)}n$. On the other hand, it is clear that $\limsup_{n}\frac{f_{10}(n)}{n}\le\lim_{n}\frac{f(n)}{n}$. Hence, $\lim_{n}\frac{f(n)}{n}=\lim_{n}\frac{f_{10}(n)}{n}$. Similar arguments hold for $f_{101}(n)$. \end{proof}

\section{Open Problems}\label{sec:conc}

We now describe some of the open problems related to extremal values of duplication distance to the root. First, the binary duplication constant, $\lim_{n}\frac{f(n)}{n}$ is unknown. It is also interesting to find bounds tighter than the one given in Theorem~\ref{thm:bounds}, namely $0.045\le\lim\frac{f(n)}{n}\le0.4$. Furthermore, although the lower bound $f(\seq s)\ge0.045n$ is valid for all but an exponentially small fraction of sequences of length $n$, we have not been able to find an explicit family of sequences whose distance is linear in $n$. A related problem to identifying sequences with large duplication distance is improving bounds on $f(\seq s)$ that depend on the structure of $\seq s$, such as the bound given in Lemma~\ref{lem:SeqDepBound}, relating $f(\seq s)$ to the number of distinct $k$-mers of $\seq s$.

While we showed in our study of approximate duplication that at $\beta=1/2$, $f_\beta(n)$ transitions from a linear dependence on $n$ to a logarithmic one, the behavior at $\beta=1/2$ is not known. Furthermore, we can alter the setting by decoupling duplications and substitutions, i.e., we generate the sequence through exact duplications and substitutions, possibly with limitations on the number of substitutions. We can then study the same problems as the ones we have in this paper as well as new problems, e.g., the minimum number substitutions required to generate the sequence via a logarithmic number of duplication steps.

A different strand of problems are algorithmic in nature, including designing an algorithm that can efficiently find or approximate the duplication distance to the root and provide a duplication process of the appropriate length. The computational complexity of these tasks is also not known. Similar questions may be asked for approximate duplication, or duplication along with substitution. These problems are important because of their potential application in determining the sequence of duplications and point mutations that may have resulted in a particular segment of an organism's DNA. 

\section*{Acknowledgment}
This work was supported in part by the NSF Expeditions in Computing
Program (The Molecular Programming Project), by a USA-Israeli
BSF grant 2012/107, by an ISF grant 620/13, and
by the Israeli I-Core program. 
\bibliographystyle{siam}
\bibliography{bib}

\end{document}